\documentclass[journal]{IEEEtran}
\usepackage{graphicx}
\usepackage{hyperref}
\usepackage{multirow}
\usepackage{array}
\usepackage{booktabs}
\usepackage{dcolumn}
\usepackage{amsmath}
\usepackage{mathtools}
\usepackage{autobreak}
\usepackage{amssymb} 
\usepackage{amsfonts}
\usepackage{cite}
\usepackage{caption}
\usepackage{hhline}
\usepackage[font=small]{caption}

\usepackage{adjustbox}
\usepackage{flexisym}
\usepackage{setspace}

\usepackage{textcomp}

\usepackage{fancyhdr}

\usepackage[utf8]{inputenc}
\usepackage[english]{babel}

\usepackage{amsthm}

\newtheorem{theorem}{Theorem}
\newtheorem{lemma}[theorem]{Lemma}

\theoremstyle{definition}

\usepackage{subfigure}
\usepackage{mathrsfs}

\usepackage{bm}

\usepackage{url}
\usepackage[noend]{algpseudocode}

\usepackage[ruled,vlined,linesnumbered]{algorithm2e}
\usepackage{varioref}

\usepackage{titlesec}
\usepackage{lipsum}

\labelformat{algocf}{{Alg.}\,#1}

\DeclareMathOperator{\sgn}{sign}

\DeclareMathOperator{\f}{f}
\DeclareMathOperator{\g}{g}

\DeclareMathOperator{\Fc}{F}
\DeclareMathOperator{\R}{R}

\newcommand{\be}{\mathbf{e}}

\newcommand{\bg}{\mathbf{g}}

\newcommand{\bq}{\mathbf{q}}
\newcommand{\br}{\mathbf{r}}
\newcommand{\bs}{\mathbf{s}}

\newcommand{\bu}{\mathbf{u}}
\newcommand{\bv}{\mathbf{v}}

\newcommand{\bx}{\mathbf{x}}
\newcommand{\by}{\mathbf{y}}
\newcommand{\bz}{\mathbf{z}}

\newcommand{\bG}{\mathbf{G}}

\newcommand{\bW}{\mathbf{W}}

\ifCLASSINFOpdf

\else

\fi

\begin{document}

\title{A Novel Key Generation Scheme Using Quaternary PUF Responses and Wiretap Polar Coding}

\author{Yonghong~Bai,~\IEEEmembership{Student~Member,~IEEE,}
Zhiyuan~Yan,~\IEEEmembership{Senior~Member,~IEEE}


\thanks{Yonghong Bai and Zhiyuan Yan are with the Department
of Electrical and Computer Engineering, Lehigh University, Bethlehem,
PA, 18015 USA (e-mail: yob216@lehigh.edu; zhy6@lehigh.edu).
}

}

\markboth{}%
{Shell \MakeLowercase{\textit{et al.}}: Bare Demo of IEEEtran.cls for IEEE Journals}

\maketitle


\begin{abstract} Physical unclonable functions (PUFs) are widely considered in secret key generation for resource constrained devices. However, PUFs require additional hardware overhead. In this paper, we focus on developing a PUF-efficient, robust, and secure key generation scheme. First, a novel method for extracting quaternary PUF responses is proposed to increase the entropy of a PUF response, in which a 2-bit response is extracted from evaluating a single PUF cell multiple times. The probability masses of the responses can be adjusted by setting parameters appropriately. Then, a chosen secret model based fuzzy extractor (FE) is designed to extract secret keys from the quaternary PUF responses. To improve the security of this FE, it is modeled as a wiretap channel system, and wiretap polar coding is adopted to reduce secrecy leakage. An upper bound of secrecy leakage is also given in this paper, and it suggests that an arbitrarily small (even zero) leakage can be achieved by properly choosing parameters of the quaternary PUF responses generation. Comparison results show that the required number of PUF cells to achieve the same level of secrecy in our scheme is as low as half that of the state-of-the-art schemes.  




\end{abstract}


\IEEEpeerreviewmaketitle

\section{Introduction} 
\label{Introduction}

\IEEEPARstart{P}{hysical} unclonable functions (PUFs) are used in secret key generations \cite{delvaux2014helper}  \cite{suzuki2018quaternary}\cite{maes2015secure}\cite{ueno2019tackling}. PUF outputs depend on random physical factors introduced in manufacturing process and environmental noise. These factors make the output of PUFs unpredictable and unclonable. To use PUFs in key generation, two important issues need to be addressed: not precisely reproducible and biased PUF outputs.


Fuzzy extractors (FEs) \cite{dodis2004fuzzy} are designed to derive keys from unstable and biased PUF bits. As shown in Fig. \ref{fz}, a secret seed is a random string that is used to derive an enrolled key through the key derivation function (KDF). Firstly, the FE generates a helper data using a PUF string and the codeword based on the secret seed, and then a legal user reconstructs the secret seed using the helper data and a noisy PUF string. Since the information of the secret seed is leaked from the helper data if the PUF bits are biased \cite{delvaux2014helper}\cite{maes2015secure}, debiasing was proposed to mitigate or reduce the secrecy leakage of FEs, such as Von Neumann correctors \cite{maes2015secure} and biased-masking \cite{ueno2019tackling}. In these PUF based key generation schemes, one PUF cell returns only one bit data upon an evaluation. If the PUF bits are not reusable for multi-enrollment, the number of PUF cells increases with the number of generated keys, which is undesirable for resource constrained devices. Therefore, a PUF-efficient key generation scheme is very important. A direct way to increase the PUF efficiency is that more bit data are derived from a single PUF cell. A method of the quaternary PUF response derivation is introduced in \cite{suzuki2018quaternary}. However, the quaternary responses are used for debiasing, rather than generating keys directly. In this paper, we propose an FE in which the quaternary PUF responses are directly used for generating keys. Like the binary PUF outputs, quaternary PUF responses are not precisely reproducible and biased. Hence, we use quaternary polar codes to hide the secrecy leakage caused by the biased quaternary responses and robustly generate keys. To the best of our knowledge, it is the first non-binary fuzzy extractor that derives secret keys from quaternary PUF responses. The main contributions of this paper are as follows:

\begin{figure}
\centering
\includegraphics[width=3.3in]{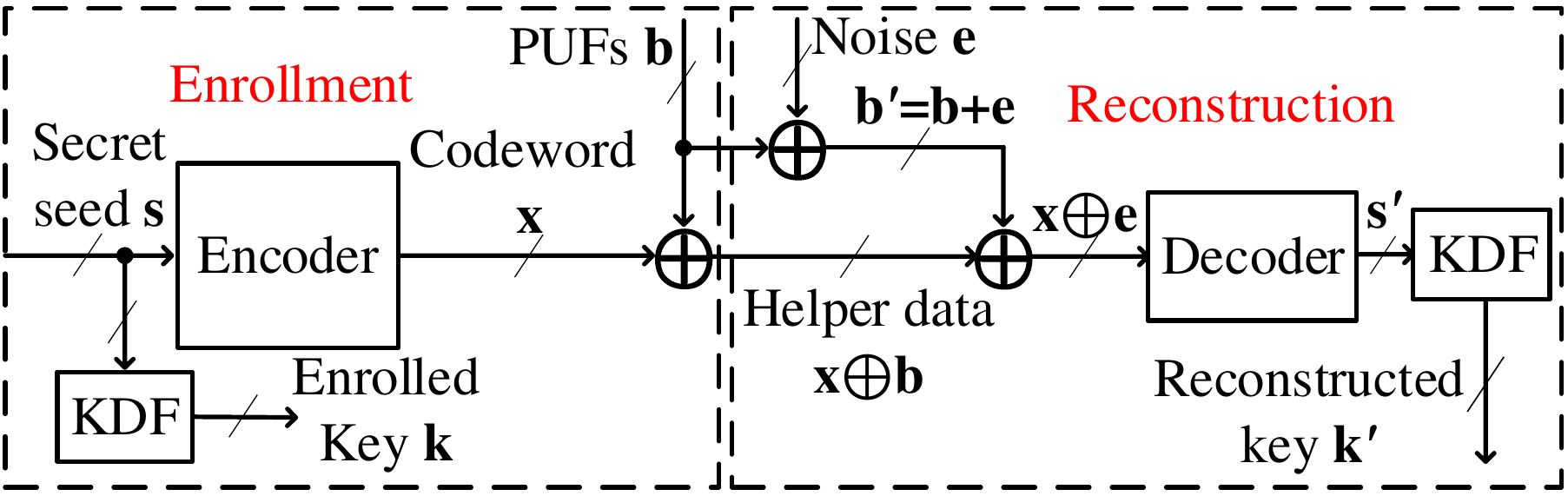}
\caption {Fuzzy extractor.}
\label{fz}
\end{figure}



\begin{enumerate}
\item A quaternary response is extracted from a single PUF cell based on the PUF model in \cite{maes2013accurate}. In this method, the probability masses of quaternary responses is adjusted by properly setting  parameters.    

\item A chosen secret model based FE is designed using quaternary PUF responses and modeled as a wiretap channel system to analyze its security. 

\item  Wiretap polar coding is adopted in the FE. A new upper bound on secrecy leakage is derived. The upper bound becomes tighter when the mask length increases, and can be made arbitrarily small (even zero) by properly choosing parameters of the quaternary PUF responses generation.
\end{enumerate}
Comparison results show that the required number of PUF cells to achieve the same level of secrecy in our scheme is as low as half that of the state-of-the-art schemes.

\section{Reviews}
\label {Reviews}

\subsection{Probabilistic reliability model of PUFs}

To describe probabilistic behavior of a PUF cell, it is sufficient to use {\it one-probability} (a random variable), which is the probability that the PUF cell returns `$1$' upon a random evaluation \cite{maes2013accurate}. The {\it one-probabilities} of different PUF cells are independent and identically distributed random variables with probability density function (PDF) $\f(x) = \frac{\lambda_1 \varphi(\lambda_2 - \lambda_1 \Phi^{-1}(x))}{\varphi (\Phi^{-1} (x))}$ and cumulative distribution function (CDF) \mbox{$\Fc(x) = \Phi(\lambda_1 \Phi^{-1}(x) - \lambda_2)$}, where $\Phi$ is the standard normal cumulative distribution function, $\varphi$ is the standard normal probability density function with two parameters $\lambda_1$ and $\lambda_2$ (see Fig. \ref{pdf} for an example).

\begin{figure}
\centering
\includegraphics[width=3.3in]{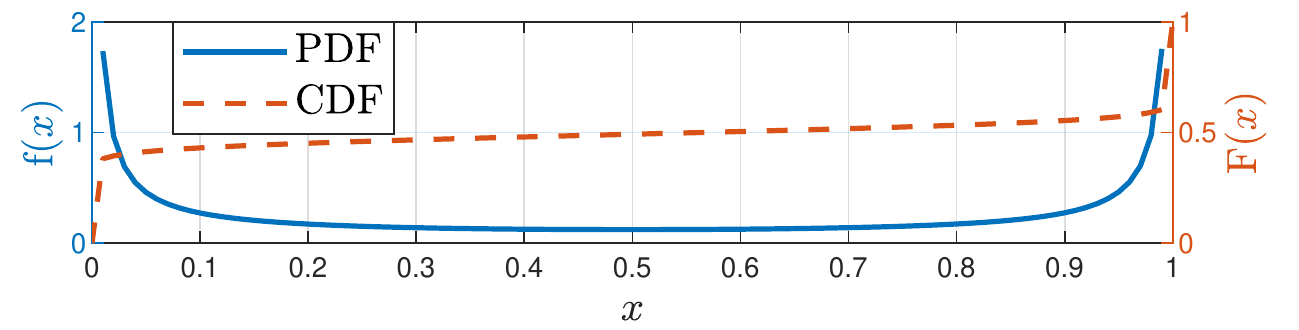}
\caption {The probability density function and cumulative distribution function of {\it one-probability} with $\lambda_1 = 0.1213$ and $\lambda_2 = 0.0210$.}
\label{pdf}
\end{figure}



\subsection{Wiretap polar coding} 
In polar codes, $N$ independent and identical DMCs are used to generalize $N$ synthesized channels whose capacities are polarized \cite{arikan2009channel}. The synthesized channels with high capacities are called good channels, while others with relatively low capacities are called bad channels. Information symbols are sent through the good channels, and hence the receiver can recover the information with high probability. The generator matrix of a polar code with length $N$ is an $N \times N$ matrix $\bG$ \cite{cheng2016encoder}. Let $\bu=\begin{bmatrix}u_0&u_1&\cdots &u_{N-1}\end{bmatrix}$ be a source word, and its codeword $\bx=\bu \bG$.

In a wiretap channel, Alice sends messages to Bob through main channels, while an eavesdropper Eve receives the messages from wiretap channels. Wiretap polar coding is proposed to protect the communications between Alice and Bob in \cite{torfi2017polar}. If the wiretap channel is stochastically degraded with respect to the main channel, Bob can recover the sent messages with high probability, whereas the probability that Eve can recover the message approaches zero \cite{torfi2017polar}. Let the transition matrices of the wiretap and main channels be $\bW_w$ and $\bW_m$, respectively. If there is a channel with a transition matrix $\bW_3$ that holds $\bW_w = \bW_3 \bW_m$, the wiretap channel will be stochastically degraded with respect to the main channel so that the channels that are good for Eve are good for Bob too and the channels that are bad for Bob are bad for Eve too \cite{torfi2017polar}. Hence, in the wiretap polar coding, indices $\{0,1, \cdots, N-1\}$ is divided into three disjoint subsets:
\begin{itemize}
  \item[] $\mathcal{R}$: Indices of the channels good for both Bob and Eve,
  \item[] $\mathcal{S}$: Indices of the channels good for Bob but bad for Eve,
  \item[] $\mathcal{F}$: Indices of the channels bad for both Bob and Eve.
\end{itemize}

The source word $\bu$ of polar codes is also partitioned into three parts: mask $\br$, secret message $\bs$, and frozen symbols $\mathbf{f}$, and they are sent through the synthesized channels that are index by $\mathcal{R}$, $\mathcal{S}$, and $\mathcal{F}$, respectively. Mask symbols are sent through the channels that are good for both Bob and Eve, and hence Eve would recover the mask symbols with high probability. However, the mask symbols are useless random symbols. The secret message is sent through the channels good for Bob but bad for Eve, and hence only Bob can recover the secret message with high probability.

\section{Quaternary fuzzy extractors} 
\label {Quaternary_FE}

In this section, we will introduce our PUF-efficient key generation scheme. First, we present the quaternary PUF responses generation method, and then we give the quaternary fuzzy extractor setup and its corresponding wiretap channel model. 

\subsection{Extraction of quaternary PUF responses}  
Similar to the method in \cite{suzuki2018quaternary}, the extraction of the quaternary PUF response of a PUF cell is based on its {\it one-frequency} (the proportion of `$1$'s in multiple PUF cell evaluations). The quaternary response of the $i$-th PUF cell
\begin{align*}
\small
q_i(f_i) = 
\begin{cases} 
0,\mbox{ if } f_i \in [0,a), \\
1,\mbox{ if } f_i \in [a,b), \\
2,\mbox{ if } f_i \in [b,c), \\
3,\mbox{ if } f_i \in [c,1], 
\end{cases}
\end{align*} where $0$, $1$, $2$, and $3$ are four elements of Galois field GF(4), $a < b < c$ are parameters, and $f_i$ is the {\it one-frequency} of the \mbox{$i$-th} PUF cell. Let $p_0$, $p_1$, $p_2$, and $p_3$ be the probabilities of the quaternary PUF response being $0$, $1$, $2$, and $3$, respectively. As the number of trials increases towards infinity, the {\it one-frequency} approaches the {\it one-probability} of a PUF cell, and hence, if $a=\Fc^{-1}(\frac{1}{4})$, $b=\Fc^{-1}(\frac{1}{2})$, and $c=\Fc^{-1}(\frac{3}{4})$, then $p_0=p_1=p_2=p_3= \frac{1}{4}$. In practice, the number of trials is not infinity. If the number of trials is small, it is hard to differentiate two close {\it one-probability} values. The larger the number of trials, the closer {\it one-frequency} and {\it one-probability} are. The small mean square errors between the experimental and theoretical error counts for a moderate number of evaluations in \cite{maes2013accurate} demonstrate, if indirectly, that the difference between {\it one-frequency} and {\it one-probability} is small with sufficient number of evaluations. To save time, PUF cells can be read offline to accumulate the `$1$'s for evaluating their {\it one-frequencies}. 

\subsection{Quaternary fuzzy extractor and wiretap channel model}
Our quaternary FE is based on well studied chosen secret model, which is shown in Fig. \ref{key_generation_to_wiretap_channel}. In the enrollment phase, a random quaternary string $\bs$ is chosen to be a secret seed, and then an enrolled key is extracted from $\bs$ using a KDF. A helper data is generated to help the reconstruction of the key in the future. First, $N$ quaternary PUF responses $\bq$ are extracted from $N$ PUF cells, and then they are added (symbol-wise addition in GF(4)) to the quaternary codeword $\bx$ of $\bs$ to get a helper data $\bx+\bq$. The helper data is published, and we assume attackers can access it. In the reconstruction phase, $N$ new quaternary PUF responses $\bq '$ are extracted from the same PUF cells and they are added to the helper data to generate $\bx+\bq + \bq'=\bx+\be$. Then, using $\bx+\be$ the polar decoder generates the decoded secret seed to reconstruct the key.    

 \begin{figure*}[!t]
\centering
\includegraphics[width=6.5in]{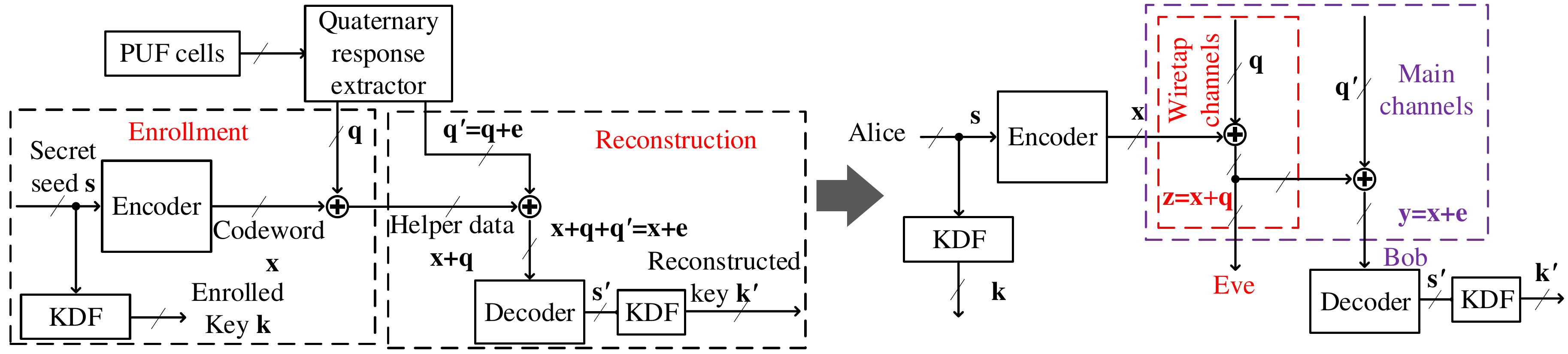}
\caption {A quaternary fuzzy extractor and its wiretap channel model. }
\label{key_generation_to_wiretap_channel}
\end{figure*}


The quaternary FE is modeled as a wiretap channel system as shown in Fig. \ref{key_generation_to_wiretap_channel}. The codeword $\bx$ is the message sent by Alice and $\by=\bx+\be$ is the received message of Bob from $N$ independent main channels, where $\be=\bq + \bq'$ is the noise vector of the $N$ main channels. Ideally $\bq=\bq'$. However,  since PUF bits are not precisely reproducible, $\bq$ and $\bq'$ could be different. A $4 \times 4$ transition probability matrix
\begin{align*}
\small
\bW_m =
\begin{bmatrix}
t_{00} & t_{01} & t_{02} & t_{03}\\
t_{10} & t_{11} & t_{12} & t_{13}\\
t_{20} & t_{21} & t_{22} & t_{23}\\
t_{30} & t_{31} & t_{32} & t_{33}\\
\end{bmatrix}
\end{align*} consists of the transition probabilities of the main channel, where $t_{jk}$ for $j\in \{0,1,2,3\}$ and $k\in \{0,1,2,3\}$ is the transition probability from input $j$ to output $k$. Similarly, the helper data $\bx+\bq$ is treated as the outputs of $N$ independent wiretap channels, where $\bq$ is the noise vector of the $N$ wiretap channels. Hence, the transition probability matrix of the wiretap channel is 
\begin{align*}
\small
\nonumber
\bW_w =
\begin{bmatrix}
p_0 & p_1 & p_2 & p_3\\
p_1 & p_0 & p_3 & p_2\\
p_2 & p_3 & p_0 & p_1\\
p_3 & p_2 & p_1 & p_0\\
\end{bmatrix}.
\end{align*}

To evaluate $\bW_m$, we first use the CDF of the {\it one-probability} to generate the {\it one-probabilities} of multiple cells. Second, we extract two quaternary responses from every PUF cell. Then, we use the numbers of quaternary values to estimate the transition probability. For example, if there are $t$ PUF cells whose original quaternary responses are $1$ and the number of $0$, $1$, $2$, and $3$ among these $t$ cells in the second quaternary responses is $w$, $x$, $y$, and $z$, respectively, then the transition rates $r_{10}=\frac{w}{t}$, $r_{11}=\frac{x}{t}$, $r_{12}=\frac{y}{t}$, and $r_{13}=\frac{z}{t}$, where $r_{ij}$ for $i\in \{0,1,2,3\}$ and $j\in \{0,1,2,3\}$ is the transition rate from input $i$ to output $j$. Our simulation results show that the rate that the quaternary response does not change is very high. We use the transition rates to approximate the transition probabilities in our simulations. For the wiretap channel, when $p_0 = p_1 = p_2 = p_3=\frac{1}{4}$, the wiretap channel is completely noisy. In this case, as long as $\bW_3$ is a $4 \times 4$ matrix whose each element is $\frac{1}{4}$, then $\bW_w = \bW_3 \bW_m$ holds so that the wiretap channel will be stochastically degraded with respect to the main channel.

As long as the wiretap channel is stochastically degraded with respect to the main channel, the non-binary wiretap polar codes of \cite{torfi2017polar} can be applied in our fuzzy extractor. The genie-aided simulation of \cite{cheng2016encoder} is used to construct the polar codes. We get the error rate of each channel through the genie-aided simulation and order the channels with increasing error rates. The first $r$ (length of mask) indices from the list for mask, the next $s$ (length of secret seed) indices for secret seed, and the remaining indices are for frozen symbols. Codewords are generated based on the recursive encoding algorithm in\cite{cheng2016encoder}. In the reconstruction phase, firstly we compute the initial LLRs for the channel outputs and then send them to the list decoder \cite{cheng2016encoder} to generate the estimated secret seed to reconstruct the enrolled key. More details of the encoding and decoding algorithms can be found in \cite{cheng2016encoder}. 


\section{Analysis of secrecy leakage}
\label{QuaterPolar}


The wiretap polar coding in \cite{torfi2017polar} is used in our FE to prevent eavesdroppers from recovering secret seed $\bs$. The $N \times N$ generator matrix $\bG$ of a quaternary polar code with length $N$ is written as $\begin{bmatrix}
\bG_{\mathcal{R}} & \bG_{\mathcal{S}} & \bG_{\mathcal{F}}
\end{bmatrix}^T$, where $\bG_{\mathcal{R}}$,  $\bG_{\mathcal{S}}$, and $\bG_{\mathcal{F}}$ consist of rows of $\bG$ that are indexed by $\mathcal{R}$, $\mathcal{S}$, and  $\mathcal{F}$, respectively. $\begin{bmatrix}
\bG_{\mathcal{R}} & \bG_{\mathcal{S}}
\end{bmatrix}^T$ and $\bG_{\mathcal{R}}$ will be the generator matrix of an ($N, r+s$) linear code $\mathcal{C}_1$ and an ($N,r$) linear code $\mathcal{C}_2$, respectively, where $r$ and $s$ are the size of $\mathcal{R}$ and $\mathcal{S}$, respectively. Hence, $\mathcal{C}_2 \subseteq \mathcal{C}_1$. In the following theorem, we give the upper bound on the secrecy leakage of our key generation scheme.

\begin{theorem} In our key generation scheme, if $p_1 =p_2=p_3=\frac{1-p_0}{3}$ and $p_0 \in \left[\frac{1}{4},1\right]$, the upper bound of secrecy leakage is
{\small
$\log_2 \left[\frac{4^N}{|\mathcal{C}_2|} \sum_{\bv \in \mathcal{C}_2} {p_0}^{N-w(\bv)}\left(\frac{1-p_0}{3}\right)^{w(\bv)}\right]$}, where $|\mathcal{C}_2|$ is the size of $\mathcal{C}_2$ and $w(\bv)$ is the Hamming weight of $\bv$.
\label{thm}
\end{theorem}

Theorem \ref{thm} is proved in the appendix. If we add one more mask symbol, the row ($\bg$) of $\bG$ indexed by the position of the added mask symbol is added to the generator matrix of $\mathcal{C}_2$. If $\bg \in \mathcal{C}_2$, then $\mathcal{C}_2$ does not change. Else, the dimension of $\mathcal{C}_2$ increases by one and $\mathcal{C}_2$ becomes $\mathcal{C}_2'=\{\mathcal{C}_2,\bg+\mathcal{C}_2\}$. Since $\bg \notin \mathcal{C}_2$, then $\bg + \mathcal{C}_2$ will be the proper coset of $\mathcal{C}_2$ and \mbox{$\frac{\Fc_{\bg+\mathcal{C}_2}(p_0)}{\Fc_{\mathcal{C}_2}(p_0)} \le 1$} for all $p_0 \in [\frac{1}{4},1]$ based on Theorem 1.19 in \cite{Klove2007codes}, where $\Fc_{\mathcal{C}}(p_0) = \frac{1}{|\mathcal{C}|} \sum_{\bv \in \mathcal{C}} {p_0}^{N-w(\bv)}(\frac{1-p_0}{3})^{w(\bv)}$. Hence, the upper bound of secrecy leakage will be $\log_2 \left[4^N \Fc_{\mathcal{C}_2'}(p_0)\right] = \log_2 \left[4^N \frac{\Fc_{\mathcal{C}_2}(p_0)+\Fc_{\bg+\mathcal{C}_2}(p_0)}{2}\right] \le \log_2 \left[4^N \Fc_{\mathcal{C}_2}(p_0)\right]$, which means that the upper bound on the secrecy leakage can be reduced by increasing the mask length. If $p_1 =p_2=p_3=\frac{1-p_0}{3}$, the secrecy leakage is reduced to a negligible value by increasing the mask length. In addition, \mbox{Theorem \ref{thm}} suggests that when $p_0 = p_1= p_2= p_3 = \frac{1}{4}$, the secrecy leakage $\le \log_2 \left[\frac{4^N}{|\mathcal{C}_2|} \sum_{\bv \in \mathcal{C}_2} {\frac{1}{4}}^{N-w(\bv)}\frac{1}{4}^{w(\bv)}\right] = 0$.




\newcommand{\tabincell}[2]{\begin{tabular}{@{}#1@{}}#2\end{tabular}}
\begin{table*}[h]
\footnotesize
\renewcommand{\arraystretch}{1.3}
\caption{Comparisons of the key generation methods. $p_0$, $p_1$, $p_2$, and $p_3$ are the probabilities of the quaternary PUF response being $0$, $1$, $2$, and $3$, respectively. For a fair comparison, all schemes are designed to generate a $128$-bit key with failure probability $\le 10^{-6}$.}
\centering
\begin{tabular}{  |p{1.6cm}< {\centering} | p{3.4cm}< {\centering} |  p{1.5cm}<{\centering} | p{0.8cm}< {\centering} | p{2.2cm}<{\centering}| p{1.5cm}< {\centering}| p{0.8cm}< {\centering} | p{2.5cm}<{\centering}|}
\hline
{}& {Quaternary channel parameters}&{PUF entropy}  & {Mask} & {Secrecy leakage }& {Code} &  {Block}& {Failure probability} \\ \hline

{\tabincell{c}{State-of-the-art \\ key generations}} &  {-}   & {{\tabincell{c}{ $1$ \cite{delvaux2014helper} \cite{maes2015secure} \cite{hiller2017hiding}, \\ $ \in [1,2]$  \cite{suzuki2018quaternary}  }}} & {-} &  {{ \tabincell{c}{ $0$  \cite{suzuki2018quaternary} \cite{maes2015secure} \cite{ueno2019tackling}, \\ negligible \cite{hiller2017hiding}\cite{bai2019secure} }}  \centering} & {\tabincell{c}{ Golay \cite{maes2015secure} \\ RM \cite{hiller2017hiding} \cite{ueno2019tackling}} } &  {-} &  {{$\le 10^{-6}$  \cite{suzuki2018quaternary} \cite{ueno2019tackling} \cite{hiller2017hiding} }}\\ \hline

{\multirow{8}{*} { This work}} & {$p_0$=$p_1$=$p_2$=$p_3$=$0.25$}  & {$2$} & {$0$} & {$0$ }& {\multirow{8}{*} {\tabincell{c}{($256,64$) qua-\\ternary polar} }} &  {$1$} & {\multirow{8}{*} {$1.45 \times 10^{-7}$}} \\ \cline{2-5} \cline{7-7}

{\multirow{8}{*} {} } & {$p_0$=$0.265$, $p_1$=$p_2$=$p_3$=$0.245$}  & {$1.9991$} & {$15$} & {$\le 0.0462$} &{\multirow{8}{*} {} }&  {$2$} & {\multirow{8}{*} {} } \\ \cline{2-5} \cline{7-7}

{\multirow{8}{*} {} } & {$p_0$=$0.268$, $p_1$=$p_2$=$p_3$=$0.244$}  & {$1.9988$} & {$15$} & {$\le 0.1092$} &{\multirow{8}{*} {} }&  {$2$} & {\multirow{8}{*} {} } \\ \cline{2-5} \cline{7-7}

{\multirow{8}{*} {} } & {$p_0$=$0.271$, $p_1$=$p_2$=$p_3$=$0.243$}  & {$1.9983$} & {$15$} & {$\le 0.2423$} &{\multirow{8}{*} {} }&  {$2$} & {\multirow{8}{*} {} } \\ \cline{2-5} \cline{7-7}

{\multirow{8}{*} {} } & {$p_0$=$0.274$, $p_1$=$p_2$=$p_3$=$0.242$}  & {$1.9978$} & {$15$} & {$\le 0.5232$} &{\multirow{8}{*} {} }&  {$2$} & {\multirow{8}{*} {} } \\ \cline{2-5} \cline{7-7}

{\multirow{8}{*} {} } & {$p_0$=$0.277$, $p_1$=$p_2$=$p_3$=$0.241$}  & {$1.9973$} & {$15$} & {$\le 1.1032$} &{\multirow{8}{*} {} }&  {$2$} & {\multirow{8}{*} {} }  \\ \cline{2-5} \cline{7-7}

{\multirow{8}{*} {} } & {\multirow{2}{*} {$p_0$=$0.28$, $p_1$=$p_2$=$p_3$=$0.24$}} &  {\multirow{2}{*} {$1.9966$}}  & {$15$} & {$\le 2.1658$} &{\multirow{8}{*} {} }&  {$2$} & {\multirow{8}{*} {} } \\ \cline{4-5} \cline{7-7}

{\multirow{8}{*} {} } &{\multirow{2}{*} {} }  &  {\multirow{2}{*} {}} & {$16$} & {$\le 0.9827$} &{\multirow{8}{*} {} }&  {$2$} & {\multirow{8}{*} {} } \\ \hline
\end{tabular}
\label{compare}
\end{table*}

\section{Simulations and comparisons}
\label{SimCompar}


We compare our key generation scheme with other state-of-the-art schemes in Table \ref{compare}.

In our scheme, when $p_0 = p_1= p_2= p_3 = \frac{1}{4}$, the entropy of the quaternary PUF response will be two, which is twice that of other key generation schemes in \cite{delvaux2014helper}\cite{maes2015secure}\cite{hiller2017hiding}. Hence,  the required number of PUF cells in our scheme is as low as half that of the schemes in \cite{delvaux2014helper}\cite{maes2015secure}\cite{hiller2017hiding}. If $p_0$, $p_1$, $p_2$, and $p_3$ are not exactly $\frac{1}{4}$ (some examples are given in Table \ref{compare}), the entropy will be less than $2$.  

Some key generation schemes with debiasing, such as Von Neumann correctors \cite{maes2015secure} and biased-masking \cite{ueno2019tackling}, have zero leakage, and some key generation schemes \cite{bai2019secure}\cite{hiller2017hiding} have negligible secrecy leakage. When $p_0 = p_1= p_2= p_3 = \frac{1}{4}$, our scheme has zero leakage without mask. When the quaternary PUF responses are biased ($p_0$, $p_1$, $p_2$, and $p_3$ are not exactly $\frac{1}{4}$), we generate all codewords of $\mathcal{C}_2$ by brute force and computed the upper bound in Theorem \ref{thm}. The simulation results are shown in Table \ref{compare}. With the same mask length, the upper bound of information leakage increases with bias. The simulation results in Table \ref{compare} show that the upper bound of secrecy leakage can be reduced by increasing the mask length as discussed in Section \ref{QuaterPolar}. In theory we could make the leakage smaller by increasing the mask length.

In our simulation, we aim to generate $128$-bit key with failure probability $\le 10^{-6}$. We consider SRAM-PUFs with $\lambda_1 = 0.1213$ and $\lambda_2 = 0.0210$ \cite{maes2013accurate}, whose average bit error rate (ABER) is $3.85\%$. Under such an ABER, the PUF-based key generation schemes in \cite{suzuki2018quaternary} \cite{ueno2019tackling} \cite{hiller2017hiding} can reconstruct a secret key with failure probability $\le 10^{-6}$. A (256, 64) quaternary polar code is used in our simulation. The list decoding \cite{cheng2016encoder} with four lists is used in our simulation. When $p_0 = p_1= p_2= p_3 = \frac{1}{4}$, only one block is enough to generate $128$-bit key. When $p_0=0.28$ and $p_1=p_2=p_3=0.24$, $16$ symbols are used as mask and each block can generate $94$-bit key. Hence, two blocks are needed in this case. The failure probability of our key reconstruction is $1.45 \times 10^{-7}$.

Many error correction codes have been considered in PUF based key generations, such as Golay codes \cite{maes2015secure}, Reed-Muller (RM) codes \cite{hiller2017hiding} \cite{ueno2019tackling}, and binary polar codes \cite{bai2019secure}. The complexity of these codes is well studied. Hence, by comparing the complexity of binary and quaternary polar decoders, we indirectly compare the complexity of our scheme with that of other state-of-the-art schemes. In the following, we provide a $\mathbf{rough}$ comparison of the computational complexities for binary and quaternary polar decoders. In a successive cancellation list (SCL) decoder, $L$ decoding paths are considered concurrently at each decoding stage and the most likely one is selected. For binary and quaternary SCL decoders, expect for LLR-updates, the other computations are similar. Hence, we focus on the LLR-updates in the following comparison. 


In a binary polar code decoder, the following two functions are used to update the LLRs of a pair of bits \cite{leroux2011hardware}
\begin{align*}
&\f\left(\lambda_0, \lambda_1\right)=\sgn ( \lambda_0)\sgn(\lambda_1)\min\left(\mid \lambda_0 \mid, \mid \lambda_1 \mid \right), \\  
&\g\left(\lambda_0,\lambda_1,\gamma \right)=\left(1-2\gamma \right) \lambda_0+\lambda_1,
\end{align*} where $\lambda_0$ and $\lambda_1$ are two inputs LLRs of the two bits, $\sgn(\lambda_0)$ returns the sign of $\lambda_0$, $\min\left \{\mid \lambda_0 \mid, \mid \lambda_1 \mid \right \}$ returns the minimum value of $\mid\lambda_0\mid$ and $\mid\lambda_1\mid$, and $\gamma$ is the partial sum of previously decoded bits. When the code length is $N$ bits, $\f$ and $g$ functions are called $\frac{N}{2}\log_2 N$ times in a decoding process, respectively. As shown in Table \ref{decoding_complexity}, There are five basic operations in $\f$: two absolute value operations, a minimum operation, and two sign operations. In the $\g$ function, if $\gamma = 0$, its output will be $\lambda_0+\lambda_1$, otherwise its output will be $\lambda_1-\lambda_0$. Hence, a $\g$ function needs one multiplexer, one addition, and one subtraction.


\begin{table*}[h]
\footnotesize
\renewcommand{\arraystretch}{1.3}
\caption{Computational complexity of binary and quaternary polar decoders. The code lengths of the binary and quaternary polar codes are $N$ bits ($N =2 \times 4^n$, $n=1,2,3, \cdots$).}
\label{decoding_complexity}
\centering
\begin{tabular}{  cccccccccccc}
\hline
{} & {function}& {Number}& {Mux}& {Min}& {Max-64}&{Max-16}& {Max-4}& {Add} & {Sub}& {Sign}&{Absolute} \\ \hline 
 {\multirow{2}{*} { Binary}  \centering} & {$\f$}&  {$\frac{N}{2}\log_2 N$}&  {-} &  {1}&  {-} &  {-} &  {-} &  {-} &  {-} &  {$2$} &  {$2$} \\ 
 
  {\multirow{4}{*} {}  \centering} & {$\g$}&  {$\frac{N}{2}\log_2 N$}&  {1} &  {-} &  {-} &  {-} &  {-} &  {$1$} &  {$1$} &  {-} &  {-} \\  \hline 

 {\multirow{4}{*} { Quaternary}  \centering} & {$\f_0$}&  {$\frac{N}{16}\log_2 \frac{N}{2}$}&  {-}&  {-} &  {$4$} &  {-} &  {-} &  {$768$} & {3} &  {-} &  {-} \\ 
 
  {\multirow{4}{*} {}  \centering} & {$\f_1$}&  {$\frac{N}{16}\log_2 \frac{N}{2}$}&  {-}&  {-} &  {-} &  {$4$} &  {-} &  {$192$}  & {3} &  {-} &  {-} \\ 
  {\multirow{4}{*} {}  \centering} & {$\f_2$}&  {$\frac{N}{16}\log_2 \frac{N}{2}$}&  {-}&  {-} &  {-} &  {-} &  {$4$} &  {$48$}  & {3} &  {-} &  {-} \\ 
  
  {\multirow{4}{*} {}  \centering} & {$\f_3$}&  {$\frac{N}{16}\log_2 \frac{N}{2}$}&  {-}&  {-} &  {-} &  {-} &  {-} &  {$12$}  & {3} &  {-} &  {-} \\ \hline

\end{tabular}
\label{compare}
\end{table*}

In a quaternary polar codes with Reed-Solomon kernel $G_{RS4}$ \cite{cheng2016encoder}, four symbols $\bu^3_0=\begin{bmatrix}u_0&u_1&u_2&u_3\end{bmatrix}$ are decoded together. There are four functions $\f_0$, $\f_1$, $\f_2$, and $\f_3$ to update the LLR of the first, second, third, and fourth symbols, respectively. When the code length is $N$ bits (to compare with the binary decoder we assume $N =2 \times 4^n$, $n=1,2,3, \cdots$), $\f_0$, $\f_1$, $\f_2$, and $\f_3$ are called $\frac{N}{16}\log_2 \frac{N}{2} $ times in a decoding process, respectively. In $\f_0$, the updated LLRs of four values \cite{cheng2016encoder}:
\begin{align*}
&\hat{\lambda}^{(0)}_0 = 0, \\  
&\hat{\lambda}^{(1)}_0 \approx \max\left(\R\left(\left[0,\bu^3_1\right]\right)\right)-\max\left(\R\left(\left[1,\bu^3_1\right]\right)\right),\\
&\hat{\lambda}^{(2)}_0 \approx \max\left(\R\left(\left[0,\bu^3_1\right]\right)\right)-\max\left(\R\left(\left[2,\bu^3_1\right]\right)\right),\\
&\hat{\lambda}^{(3)}_0 \approx \max\left(\R\left(\left[0,\bu^3_1\right]\right)\right)-\max\left(\R\left(\left[3,\bu^3_1\right]\right)\right),
\end{align*} where $\left[i,\bu^3_1\right]$ for $i=0,1,2,3$ is the symbol vector $\bu^3_0$ with $u_0=i$, $\R(\bu^3_0) = -\sum_{r=0}^{3} \lambda^{(x_r)}_r$, $\lambda^{(x_r)}_r$ is the input LLR of the $r$-th symbol with $x_r$ for $i=0,1,2,3$ ($\bx^3_0=\begin{bmatrix}x_0&x_1&x_2&x_3\end{bmatrix}=\bu^3_0G_{RS4}$), and $\max\left(\R\left(\left[i,\bu^3_1\right]\right)\right)$ for $i=0,1,2,3$ returns the maximum value of all $\R\left(\left[i,\bu^3_1\right]\right)$ results. There are $4^3=64$ combinations of $\bu_1^3$, and hence $\R$ function is called $256$ times in $\f_0$. Since an $\R$ function needs three additions, the number of addition operations is $768$. In addition, a maximum operation is needed to find the maximum number from the $64$ results of $\R$ (Max-64 in Table \ref{decoding_complexity}). Finally, the four LLRs of $\f_0$ is generated after three subtraction operations. Similarly, $\f_1$, $\f_2$, and $\f_3$ need $192$, $48$, and $12$ addition operations, respectively. $\f_1$ needs a maximum operation to find the maximum number from $16$ results of $\R$ (Max-16 in Table \ref{decoding_complexity}) and $\f_2$ needs a maximum operation to find the maximum number from $4$ results of $\R$ (Max-4 in Table \ref{decoding_complexity}). 

Table \ref{decoding_complexity} shows that the binary polar decoder uses only a few simple operations. However, the quaternary polar decoder uses a lot of addition and maximum operations.

\section{Conclusion}
\label{Con}

A secure, robust, and efficient fuzzy extractor based on quaternary PUF responses and wiretap polar coding is proposed in this paper. To analyze the secrecy leakage, we build the wiretap channel model for the fuzzy extractor. The wiretap polar coding is adopted in the fuzzy extractor to hide secrecy leakage and ensure the robustness of the key generation. The upper bound of the secrecy leakage is proposed in this paper, and we show that the leakage can be zero by properly setting parameters of the quaternary PUF responses generation.      

\section*{Appendix: Proof of Theorem \ref{thm}}

The secrecy leakage (SL) is the mutual information between the secret seed $\bs \in \{0,1,2,3\}^s$ and the outputs of wiretap channels $\bz \in \{0,1,2,3\}^N$
\begin{align}
\nonumber
& \mbox{SL} = I(\bs;\bz) = H(\bz) - H(\bz|\bs)\\
&=\sum_{\bz} \Pr\{\bz|\bs\} \log_2 \Pr \{\bz|\bs\} - \sum_{\bz} \Pr\{\bz\} \log_2 \Pr\{\bz\},
\label{IL}
\end{align} where $I$ is a mutual information function and $H$ is a binary entropy function. To compute $H(\bz|\bs)$, first we have 
\allowdisplaybreaks
\begin{align*}
& \Pr\{\bz|\bs_i\} =\sum_{\bx \in \bx_i + \mathcal{C}_2} \Pr\{\bx|\bs_i\} \Pr\{\bz|\bx,\bs_i\} \\
&  = \frac{1}{4^{r}} \sum_{\bx \in \bx_i + \mathcal{C}_2} \Pr\{\bz|\bx\} \\
& = \frac{1}{4^{r}} \sum_{\bx \in \bx_i + \mathcal{C}_2} p_0^{N-d(\bx,\bz)}(\frac{1-p_0}{3})^{d(\bx,\bz)}\\ 
& = \frac{1}{4^{r}} \sum_{\bv \in \bz - \bx_i - \mathcal{C}_2} p_0^{N-w(\bv)}(\frac{1-p_0}{3})^{w(\bv)}, 
\end{align*} where $\bs_i$ is the $i$-th secret message, $\bx_i= \bs_i \bG_{\mathcal{S}}$, $d$ is a Hamming distance function, and $w(\bv)$ is the Hamming weight of $\bv$. Set $\{\bz - \bx_i - \mathcal{C}_2,\bz \in \{0,1,2,3\}^N\}$ and $\{\bz - \mathcal{C}_2,\bz \in \{0,1,2,3\}^N\}$ are the same for a given $\bs_i$, and so are $\{\Pr\{\bz|\bs_i\},\bz \in \{0,1,2,3\}^N\}$ and $\{\Pr\{\bz|\bs = 0\},\bz \in \{0,1,2,3\}^N\}$. Hence, $H(\bz|\bs) = -\sum_{\bz} \Pr\{\bz|\bs=0\} \log_2 \Pr \{\bz|\bs=0\}$, where $ \Pr \{\bz|\bs=0\} = \frac{1}{4^{r}} \sum_{\bv \in \bz - \mathcal{C}_2} p_0^{N-w(\bv)}(\frac{1-p_0}{3})^{w(\bv)}$. To compute $H(\bz)$, we have 
\allowdisplaybreaks
\begin{align*}
& \Pr\{\bz\} =\sum_{i = 1}^{4^s} \Pr\{\bs_i\} \Pr\{\bz|\bs_i\} \\
& = \sum_{i = 1}^{4^s} \Pr\{\bs_i\} \sum_{\bx\in \bx_i + \mathcal{C}_2} \Pr\{\bx|\bs_i\} \Pr\{\bz|\bx,\bs_i\} \\
\nonumber
& = \sum_{i = 1}^{4^s} \frac{1}{4^{s}} \sum_{\bx\in \bx_i + \mathcal{C}_2} \frac{1}{4^{r}} \Pr\{\bz|\bx\} \\
& = \frac{1}{4^{r+s}} \sum_{\bx\in \mathcal{C}_1} p_0^{N-d(\bx,\bz)}(\frac{1-p_0}{3})^{d(\bx,\bz)} \\
& = \frac{1}{4^{r+s}} \sum_{\bv \in \bz - \mathcal{C}_1} p_0^{N-w(\bv)}(\frac{1-p_0}{3})^{w(\bv)}. 
\end{align*} 

Define a new function 
\begin{align}
\Fc_{\mathcal{C}}(p_0) = \frac{1}{|\mathcal{C}|} \sum_{\bv \in \mathcal{C}} {p_0}^{N-w(\bv)}(\frac{1-p_0}{3})^{w(\bv)}.
\label{pc}
\end{align} Hence,
\allowdisplaybreaks
\begin{align}
& \Pr\{\bz|\bs=0\} =  \Fc_{\bz-\mathcal{C}_2}(p_0), \label{pzs0}\\ 
& \Pr\{\bz\} = \Fc_{\bz-\mathcal{C}_1}(p_0).
\label{pz}
\end{align}

\begin{lemma} 
For an ($N$, $k$) code $\mathcal{C}$ and all $\bz \in \{0,1,2,3\}^N$,
$\sum_{\bz} \Fc_{\bz-\mathcal{C}}(p_0) = 1$.
\label{equal1}
\end{lemma}

\begin{proof}Let $q = \frac{1-p_0}{3}$ and $m = |\bz-\mathcal{C}|$, and $\sum_{\bz} \Fc_{\bz-\mathcal{C}}(p_0) = \frac{1}{m}[m{p_0}^{N} {q}^{0} {q}^{0} {q}^{0} + mN{p_0}^{N-1} {q}^{1} {q}^{0} {q}^{0} + m{N \choose 2}{p_0}^{N-2} {q}^{2} {q}^{0} {q}^{0} + \cdots + m{N \choose h}{N-h \choose i}{N-h-i \choose j} {N-h-i-j \choose k}{p_0}^{h} {q}^{i} {q}^{j} {q}^{k} + \cdots +m{p_0}^{0} {q}^{0} {q}^{0} {q}^{N}] = \sum \frac{N!}{h!i!j!k!} {p_0}^{h} {q}^{i} {q}^{j} {q}^{k} = (p_0+3q)^N  = 1$.
\end{proof}

Based on (\ref{pzs0}), (\ref{pz}), and Lemma \ref{equal1},  
{\footnotesize
\begin{align}
&\mbox{SL} = \sum_{\bz} \Fc_{\bz-\mathcal{C}_2}(p_0) \log_2 \Fc_{\bz-\mathcal{C}_2}(p_0) - \sum_{\bz} \Fc_{\bz-\mathcal{C}_1}(p_0) \log_2 \Fc_{\bz-\mathcal{C}_1}(p_0) \nonumber \\
&= \sum_{\bz} \Fc_{\bz-\mathcal{C}_2}(p_0) \log_2 \frac{\Fc_{\bz-\mathcal{C}_2}(p_0)}{\Fc_{\mathcal{C}_2}(p_0)}+ \sum_{\bz} \Fc_{\bz-\mathcal{C}_1}(p_0) \log_2 \frac{\Fc_{\mathcal{C}_2}(p_0)}{\Fc_{\bz-\mathcal{C}_1}(p_0)}.
\label{SL2}
\end{align}}

\begin{lemma} For code $\mathcal{C}_2$ and all $\bz \in \{0,1,2,3\}^N$, if $ \frac{1}{4} \le p_0 \le 1$, 
$\sum_{\bz} \Fc_{\bz-\mathcal{C}_2}(p_0) \log_2 \frac{\Fc_{\bz-\mathcal{C}_2}(p_0)}{\Fc_{\mathcal{C}_2}(p_0)} \le 0$.
\label{l3}
\end{lemma}

\begin{proof}
Recall the weight distribution function of a code $\mathcal{C}$ \cite{Klove2007codes} 
$A_{\mathcal{C}}(w) = \sum_{i=0}^N  n_i w^i$, where $n_i$ is number of codewords with Hamming weight $i$ ($i=0,1,\cdots,N$), then the new function in (\ref{pc}) will be $\Fc_{\mathcal{C}}(p_0) =  \frac{1}{|\mathcal{C}|} [{p_0}^N + n_1{p_0}^{N-1}(\frac{1-p_0}{3})^1+ \cdots +n_N(\frac{1-p_0}{3})^N] = \frac{1}{|\mathcal{C}|} {p_0}^N [1 + n_1(\frac{1-p_0}{3p_0})^1+ \cdots +n_N(\frac{1-p_0}{3p_0})^N] =  \frac{1}{|\mathcal{C}|} {p_0}^N A_{\mathcal{C}}(\frac{1-p_0}{3p_0})$. If $\bz \notin \mathcal{C}_2$, $\bz - \mathcal{C}_2$ will be the proper coset of $\mathcal{C}_2$, and then $\frac{\Fc_{\bz-\mathcal{C}_2}(p_0)}{\Fc_{\mathcal{C}_2}(p_0)} = \frac{A_{\bz-\mathcal{C}_2}(\frac{1-p_0}{3p_0})}{A_{\mathcal{C}_2}(\frac{1-p_0}{3p_0})} \le \frac{1-(\frac{4p_0-1}{3})^{r+1}}{1+3(\frac{4p_0-1}{3})^{r+1}} \le 1$ for all $p_0 \in [\frac{1}{4},1]$ according to the Theorem 1.19 in \cite{Klove2007codes}.
\end{proof}


\begin{lemma}
For codes $\mathcal{C}_1$, $\mathcal{C}_2$, and all $\bz \in \{0,1,2,3\}^N$,
$\sum_{\bz} \Fc_{\bz-\mathcal{C}_1}(p_0) \log_2 \frac{\Fc_{\mathcal{C}_2}(p_0)}{\Fc_{\bz-\mathcal{C}_1}(p_0)} \le  \log_2 \left[4^N \Fc_{\mathcal{C}_2}(p_0)\right].$
\label{l4}
\end{lemma}

\begin{proof} 
\begin{align*}
& \sum_{\bz} \Fc_{\bz-\mathcal{C}_1}(p_0) \log_2 \frac{\Fc_{\mathcal{C}_2}(p_0)}{\Fc_{\bz-\mathcal{C}_1}(p_0)} \\  
& = \sum_{\bz} \Fc_{\bz-\mathcal{C}_1}(p_0) \log_2 \Fc_{\mathcal{C}_2}(p_0) - \sum_{\bz} \Fc_{\bz-\mathcal{C}_1}(p_0) \log_2 \Fc_{\bz-\mathcal{C}_1}(p_0) \\
& = \log_2 \Fc_{\mathcal{C}_2}(p_0) + H(\Fc_{\bz-\mathcal{C}_1}(p_0)) \le \log_2 \Fc_{\mathcal{C}_2}(p_0) + 2N \\
& = \log_2 \left[4^N \Fc_{\mathcal{C}_2}(p_0)\right].
\end{align*}
\end{proof}
Based on (\ref{pc}), (\ref{SL2}), Lemma \ref{l3}, and Lemma \ref{l4}, we get $\mbox{SL} \le \log_2 \left[\frac{4^N}{|\mathcal{C}_2|} \sum_{\bv \in \mathcal{C}_2} {p_0}^{N-w(\bv)}\left(\frac{1-p_0}{3}\right)^{w(\bv)}\right]$ for all $p_0 \in \left[\frac{1}{4},1\right]$.

\ifCLASSOPTIONcaptionsoff
  \newpage
\fi

\bibliographystyle{IEEEtran}
\bibliography{reference}

\end{document}